\documentclass{article}
\pdfoutput=1
\title{Nested Sequents for Quasi-transitive~Modal~Logics}
\author{Sonia Marin and Paaras Padhiar}
\date{May 2024}

\usepackage{quiver}
\usepackage{tikz}
\usepackage{amsthm}
\usepackage{latexsym}
\usepackage{amsmath}

\usepackage{virginialake}

\usepackage[inline]{enumitem}
\usepackage{colonequals}
\newtheorem{theorem}{Theorem}
\newtheorem{lemma}[theorem]{Lemma}
\newtheorem{proposition}[theorem]{Proposition}
\newtheorem{corollary}[theorem]{Corollary}

\theoremstyle{definition}
\newtheorem{definition}[theorem]{Definition}
\vlnosmallleftlabels

\newcommand{\vlstr}[3]{\vltr{#1}{#2}{\vlhy{}}{#3}{\vlhy{}}}
\newcommand{\vlhtr}[2]{\vlstr{#1}{#2}{\vlhy{\hskip1.5em}}}

\newcommand{\AND}{\wedge}
\newcommand{\OR}{\vee}
\newcommand{\IMPLIES}{\supset}
\newcommand{\IMP}{\supset}
\newcommand{\DIA}{\Diamond}
\newcommand{\BAR}[1]{\bar{#1}}
\newcommand{\BOX}{\Box}

\newcommand{\K}{\mathsf{K}}

\newcommand{\seq}[2]{#1 \{ #2 \}}
\newcommand{\path}[1]{\mathsf{4}^{#1}}

\newcommand{\fax}[1]{\mathsf{4}^{#1}}

\newcommand{\form}{\text{form}}

\begin{document}

\maketitle

\begin{abstract}
Previous works by Goré, Postniece and Tiu have provided sound and cut-free complete proof systems for modal logics extended with path axioms using the formalism of nested sequent.
Our aim is to provide (i) a constructive cut-elimination procedure and (ii) alternative modular formulations for these systems. 
We present our methodology to achieve these two goals on a subclass of path axioms, namely quasi-transitivity axioms.
  \end{abstract}

\section{Introduction}
The proof theory of modal logics has been explored thoroughly and many authors have contributed to the deep understanding gathered to this day.
In particular, it has been remarked time and time again that in order to capture the validities of a modal logic, additional structure, often inspired by the semantics of the logic itself, is required within the proof-theoretical syntax.
This led to the development of many formalisms extending Gentzen's sequent calculus, such as hypersequents~\cite{avron_method_1996}, nested sequents~\cite{brunnler_deep_2009,poggiolesi_gentzen_2011}, and labelled sequents~\cite{negri_proof_2005}.

It is not always clear however what sort of additional structure is precisely required to design the proof theory of a modal logic. 
For example, modal logic $\mathsf{S5}$ can be expressed using labelled or nested sequents, but can also be given a sound and complete system in the lighter hypersequent formalism, whereas such a result is conjectured not to be possible in ordinary sequent calculus~\cite{lellmann_axioms_2014}.

Goré, Postniece and Tiu~\cite{gore_correspondence_2011} have proposed a general algorithm to design nested sequent systems for modal logic $\K$\footnote{Their work takes place in the context of tense logic where the language contains also adjoint modalities, but we restrict our attention to the language with only $\BOX$ and $\DIA$.} extended with \emph{path axioms}, of the form $\DIA^n A \IMP\BOX^l\DIA A$.
As a subclass of Scott-Lemmon axioms $\DIA^n\BOX^m A \IMP\BOX^l\DIA^k A$, they enjoy a well-behaved correspondence with the frame conditions displayed on the left of Figure~\ref{fig:frame-cond}, 
i.e., if $uR^nv$ and $uR^lw$ then $vRw$~\cite{lemmon_introduction_1977}.

In this line of work, we set out to understand more precisely the systems from~\cite{gore_correspondence_2011} proof theoretically, in particular we are after a methodology to (i) equip them with a constructive \emph{cut-elimination} procedure and (ii) distill them into \emph{modular} systems, i.e., such that each axiom corresponds to a (set of) rule(s) which can be freely mixed with others without adaptation depending on the other axioms present.

In this report, we focus on a restricted class of path axioms, which we call \emph{quasi-transitivity}, that is, when $l = 0$ and $\fax{n}: \DIA^n A \IMP \DIA A$,
which correspond to the frame conditions displayed on the right of Figure~\ref{fig:frame-cond}.
Similarly to~\cite{gore_correspondence_2011}, we work in the setting of \emph{nested sequents}~\cite{brunnler_deep_2009,poggiolesi_gentzen_2011}. 
After some preliminaries in Section~\ref{sec:prelim}, we provide a cut-elimination argument for the nested sequent systems given by~\cite{gore_correspondence_2011} in Section~\ref{sec:cutelim} and we present an alternative modular cut-free system in Section~\ref{sec:modular}.

\begin{figure*}
\centering
\[\begin{tikzcd}
	& v \\
	u &&& u \\
	& w &&& w
	\arrow["vRw", dashed, from=1-2, to=3-2]
	\arrow["{{uR^nv}}", from=2-1, to=1-2]
	\arrow["{{uR^lw}}"', from=2-1, to=3-2]
	\arrow["{{uR^nw}}"', from=2-4, to=3-5]
	\arrow["uRw", curve={height=-12pt}, dashed, from=2-4, to=3-5]
\end{tikzcd}\]
\caption{Frame conditions for path axioms and quasi-transitivity axioms}
\label{fig:frame-cond}
\end{figure*}
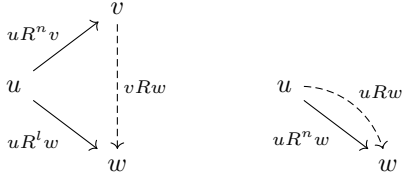

\section{Preliminaries}\label{sec:prelim}
\begin{definition}
    We denote by $\mathsf{Atm}$ a countable set of propositional variables. \emph{Formulas}, denoted $A, B, C, \dots$, are written in negation normal form, defined by the grammar:
    $$A\coloncolonequals p \in \mathsf{Atm} \mid \BAR{p} \mid (A \AND A) \mid (A \OR A) \mid \BOX A \mid \DIA A$$
    For $p\in \mathsf{Atm}$, $\BAR{p}$ denotes its negation. We define $\bot \colonequals p_0 \AND \BAR {p_0}$ and $\top \colonequals p_0 \OR \BAR{p_0}$ for a fixed $p_0 \in \textsf{Atm}$. We define the negation $\BAR{A}$ of $A$ inductively by de Morgan duality.
    $A \IMP B$ is defined as $\BAR{A} \OR B$.  We omit the outermost parentheses of a formula.
\end{definition}

\begin{definition}
    Let $A$ be a formula. The \emph{degree} of the formula $A$, $\text{deg}(A)$, is defined inductively:
    \begin{itemize}
        \item if $A = p$ or $\BAR{p}$ for some $p \in \mathsf{Atm}$, $\text{deg}(A)=0$;
        \item if $A = B \AND C$ or $B \OR C$ for some formulas $B$ and $C$, $\text{deg}(A)=\text{deg}(B)+\text{deg}(C)$;
        \item if $A = \BOX B$ or $\DIA B$ for some formula $B$, $\text{deg}(A)=1+\text{deg}(B)$.
    \end{itemize}
\end{definition}

Modal logic $\K$ is an extension of classical propositional logic with the following axiom schema, and the inference rules \emph{modus ponens} and \emph{necessitation}~\cite{blackburn_modal_2001}:
$$
\mathsf k : \BOX(A \IMP B) \IMP (\BOX A \IMP \BOX B) 
\qquad
\vliinf{\mathsf{mp}}{}{B}{A\IMP B}{A} 
\qquad
\vlinf{\mathsf{nec}}{}{\BOX A}{A}
$$

We concern ourselves with extensions of $\K$ with quasi-transitivity axioms:
$$
  \path{n}: \DIA^n A \IMP \DIA A
$$
for natural numbers $n>1$, and a nested sequent system sound and cut-free complete with these logics.

\begin{definition}
    A \emph{nested sequent} is a finite multiset of formulas and boxed sequents, that is, expressions of the form $[\Gamma]$ where $\Gamma$ is a nested sequent. In other words, a nested sequent is of the form
    $$\Gamma\coloncolonequals\varnothing \mid A,\Gamma \mid \Gamma,[\Gamma]$$
    The corresponding formula of nested sequent $\Gamma$, $\form(\Gamma)$, is defined inductively:
    \begin{itemize}
        \item $\form(\varnothing)=\bot$
        \item $\form(A,\Gamma_1)=A\OR \form(\Gamma_1)$
        \item $\form(\Gamma_1, [\Gamma_2]) = \form(\Gamma_1) \OR \BOX \form(\Gamma_2)$.
    \end{itemize}
\end{definition}

\begin{definition}
   A \emph{context} is a nested sequent with one or several holes. A hole~$\{ \ \}$ takes the place of a formula in the sequent but does not occur
inside a formula. 
We write $\Gamma \{\Gamma_1 \}$ when we replace the hole in $\Gamma \{ \ \}$ by $\Gamma_1$.
\end{definition}

\begin{definition}
    The \emph{depth} of a context $\Gamma \{ \ \}$ is defined inductively:
    \begin{itemize}
        \item If $\Gamma \{ \ \} = \Gamma_1,  \{ \ \}$ for some nested sequent $\Gamma_1$, $\text{depth}(\Gamma \{ \ \})=0$;
        \item If $\Gamma \{ \ \} = \Gamma_1, [\Gamma_2  \{ \ \}]$ for some nested sequent $\Gamma_1$ and context $\Gamma_2  \{ \ \}$, $\text{depth}(\Gamma \{ \ \})=1+\text{depth}(\Gamma_2 \{ \ \})$.
    \end{itemize}
\end{definition}

\begin{definition}
A \emph{proof} of a nested sequent $\Gamma$ is a finite tree labelled by nested sequents where $\Gamma$ is the root of the tree and:
\begin{itemize}
    \item it has no children if it is an instance of the $\mathsf{id}$ rule;
    \item or its children are $\Gamma_1, \dots, \Gamma_n$ if it is the conclusion of the inference rule
    $$\vliiinf{\rho}{}{\Gamma}{\Gamma_1}{\dots}{\Gamma_n}$$
\end{itemize}

A sequent $\Gamma$ is provable in system $\mathsf{N}$ denoted $\mathsf{N} \vdash \Gamma$ if there is a proof of it in $\mathsf{N}$.

Figure~\ref{fig:nk} gives the rules for the nested sequent system $\mathsf{nK}$.
Figure~\ref{fig:cut} is the nested sequent version of the $\mathsf{cut}$ rule and Figure~\ref{fig:path} gives the rules $\DIA_{\mathsf{k}n}$ which will be used to treat quasi-transitivity in this system, following~\cite{gore_correspondence_2011}.

For any $\mathsf{X} \subseteq  \{n \in \mathbb{N}: n>1 \}$, let us write $\path{\mathsf{X}} \colonequals \{ \path{n} : n \in \mathsf{X} \}$ and $\DIA_{\mathsf{k}\mathsf{X}} \colonequals \{ \DIA_{\mathsf{k}n} : n \in \mathsf{X} \}$.
\end{definition}

\begin{definition}
    Let $\Gamma$ be a nested sequent and $\pi$ be a proof of~$\Gamma$ in a nested sequent system. The \emph{height} of a proof, $h(\pi)$ is defined inductively:
    \begin{itemize}
        \item If $\Gamma$ is a conclusion of the $\mathsf{id}$ rule in $\pi$, then $h(\pi)=0$;
        \item If $\Gamma$ is a conclusion of an inference rule $\rho$ in $\pi$, where the proof is of the form
        $$\vlderivation
        {
            \vliiin{\rho}{}
            {\Gamma}
            {\vlhtr{\pi_1}{\Gamma_1}}
            {\vlhy{\dots}}
            {\vlhtr{\pi_n}{\Gamma_n}}
        }
        $$
        then $h(\pi)=1+\text{max}(h(\pi_1), \dots , h(\pi_n))$.
    \end{itemize}
\end{definition}

\begin{figure}
    \centering
    $
    \vlinf{\mathsf{id}}{}{\Gamma\{p,\bar{p} \}}{}
    \quad
    \vlinf{\OR}{}{\Gamma\{A \OR B \}}{\Gamma\{A, B \}}
    \quad
    \vliinf{\AND}{}{\Gamma\{A \AND B \}}{\Gamma\{A \}}{\Gamma\{B \}}
    \quad
    \vlinf{\BOX}{}{\Gamma\{\BOX A \}}{\Gamma\{[A] \}}
    \quad
    \vlinf{\DIA}{}{\Gamma\{\DIA A, [\Delta] \}}{\Gamma\{\DIA A, [A, \Delta] \}}
    $
    \caption{System $\mathsf{nK}$}
    \label{fig:nk}

    \bigskip
    \centering
    $\vliinf{\mathsf{cut}}{}{\Gamma \{ \varnothing \}}{\Gamma \{ A \}}{\Gamma \{ \BAR{A} \}}$
    \caption{Cut rule}
    \label{fig:cut}

    \bigskip
    \centering
    $\vlinf{\DIA_{\mathsf{k}n}}{}
    {\Gamma  \{ 
\DIA A, [\Delta_1, [\dots, [\Delta_{n-1},[\Delta_{n}]] \dots ]] \}}
    {\Gamma  \{ 
\DIA A, [\Delta_1, [\dots, [\Delta_{n-1},[A, \Delta_{n}]] \dots ]] \}}$
    \caption{Modal propagation rules $\DIA_{\mathsf{k}n}$ where $n\geq 1$}
    \label{fig:path}

\end{figure}

\begin{definition}
    Let $\mathsf{X} \subseteq  \{n \in \mathbb{N}: n>1 \}$. Let $\Gamma \{ \ \}$ be a context. Let $A$ be a formula. Let $r \in \mathbb{N}$. The \emph{cut-rank} of a $\mathsf{cut}$ is the degree of the \emph{cut-formula} $A$.
    When $\text{deg}(A) \leq r$, we may write $\mathsf{cut}_r$ to be explicit.
    Given a proof $\pi$ in $\mathsf{nK} + \DIA_{\mathsf{k}\mathsf{X}} + \mathsf{cut}$, the \emph{cut-rank} of $\pi$ denoted $r(\pi)$ of the proof is the supremum of the cut-ranks of cuts used in the proof. We say the proof is in $\mathsf{nK} + \DIA_{\mathsf{k}\mathsf{X}} + \mathsf{cut}_r$ when the cut-rank of the proof is less than or equal to $r$.
\end{definition}

\begin{definition}
    A rule of the form $\vliiinf{\rho}{}{\Gamma}{\Gamma_1}{\dots}{\Gamma_n}$ is (resp.~\emph{height preserving}) (resp.~\emph{cut-rank preserving}) \emph{admissible} in a nested sequent system $\mathsf{N}$ if whenever there are proofs $\pi_i$ of $\Gamma_i$ in $\mathsf N$
    for all $i \in \{ 1, \dots, n\}$, then there is a proof $\pi$ of $\Gamma$ in $\mathsf N$ (resp.~such that $h(\pi) \leq h(\pi_i)$) (resp.~such that $r(\pi) \leq r(\pi_i)$). 
\end{definition}

For each rule $\vliiinf{\rho}{}{\Gamma}{\Gamma_1}{\dots}{\Gamma_n}$, when $n>0$, its inverses $\vlinf{\rho^{-1}_i}{}{\Gamma_i}{\Gamma}$ are obtained for each $i\le n$ by ``exchanging" premiss and conclusion. The inverse rules of $\mathsf{nK}$ are shown in Figure \ref{fig:inv}. 

We adapt Br\"unnler's cut-elimination proof for nested sequent systems for logics of the modal $\mathsf{S5}$ cube \cite{brunnler_deep_2009}, similarly utilising height-preserving and cut-rank preserving admissibility of the inverse rules, as well as contraction and weakening rules. 

\begin{figure}
    \centering
    $\vlinf{\mathsf{gid}}{}{\Gamma  \{ A, \BAR{A} \}}{}$
    \caption{Generalised identity rule}
    \label{fig:gid}

    \bigskip
    \centering
    $
    \vlinf{\OR^{-1}}{}{\Gamma  \{ A, B\}}{\Gamma  \{ A \OR B\}}
    \quad
    \vlinf{\AND^{-1}_{i}}{\text{\footnotesize $i \in \{ 1, 2 \}$}}{\Gamma  \{ A_i \}}{\Gamma  \{ A_1 \AND A_2 \}}
    \quad
    \vlinf{\BOX^{-1}}{}{\Gamma  \{ [A] \}}{\Gamma  \{ \BOX A \}}
    \quad
    \vlinf{\DIA^{-1}}{}{\Gamma\{\DIA A, [A, \Delta] \}}{\Gamma\{\DIA A, [\Delta] \}}
    $
    \caption{Inverse rules for $\mathsf{nK}$}
    \label{fig:inv}

    \bigskip
    \centering
    $
    \vlinf{\mathsf{wk}}{}{\Gamma\{ \Delta \} }{\Gamma\{ \varnothing \}}
    \qquad
    \vlinf{\mathsf{cont}}{}{\Gamma\{ \Delta \} }{\Gamma\{ \Delta, \Delta \}}
    $
    \caption{Weakening and contraction rules}
    \label{fig:wkcont}
\end{figure}

\begin{proposition} \label{prop:invadm}
    Let $\mathsf{X} \subseteq  \{n \in \mathbb{N}: n>1 \}$. 
    \begin{itemize}
        \item The rule $\mathsf{gid}$ is admissible in $\mathsf{nK}+\DIA_{\mathsf{k}\mathsf{X}}$. 
        \item The inverse rules of $\mathsf{nK}+\DIA_{\mathsf{k}\mathsf{X}}+ \mathsf{cut}$, $\mathsf{wk}$ and $\mathsf{cont}$ are cut-rank and height preserving admissible in $\mathsf{nK}+\DIA_{\mathsf{k}\mathsf{X}} + \mathsf{cut}$.
    \end{itemize}
\end{proposition}

\begin{proof}
For $\mathsf{gid}$, we show the rule is admissible by induction on $\text{deg}(A)$. 

The base case $\text{deg}(A)=0$: $A = p$ or $\BAR{p}$ for some $p \in \mathsf{Atm}$, and so the rule is an instance of the $\mathsf{id}$ rule.

Assume the result holds for formulas $A$ where $\text{deg}(A) < d$ for some natural number $d > 0$.
When $\text{deg}(A)=d$, we look at cases:
\begin{enumerate}
    \item $A = \BOX B$ or $\DIA \BAR{B}$ for some formula $B$ with $\text{deg}(B) < d$. We have
    $$
    \vlderivation{
       \vlin{\BOX}{}
       {\Gamma  \{ \BOX B, \DIA \BAR{B} \}}
       {
            \vlin{\DIA}{}
            {
                \Gamma  \{  [B], \DIA \BAR{B} \}
            }
            {
                \vlin{\mathsf{gid}}{*}
                {
                \Gamma  \{  [B, \BAR{B}], \DIA \BAR{B} \}
                }
                {
                \vlhy{}
                }
            }
       }
       }
    $$ where $*$ denotes where the inductive hypothesis is used.
    \item $A = B \AND C$ or $\BAR{B} \OR \BAR{C}$ for some formulas $B$ and $C$ where $\text{deg}(B),\text{deg}(C) < d$. We have
    $$
        \vlderivation
        {
        \vlin{\OR}{}
        {
        \Gamma  \{ B \AND C, \BAR{B} \OR \BAR{C} \}
        }
        {
            \vliin{\AND}{}
            {
            \Gamma  \{ B \AND C, \BAR{B}, \BAR{C} \}
            }
            {
                \vlin{\mathsf{gid}}{*}
                {
                \Gamma  \{ B, \BAR{B}, \BAR{C} \}
                }
                {
                \vlhy{}
                }
            }
            {
                \vlin{\mathsf{gid}}{*}
                {
                \Gamma  \{ C, \BAR{B}, \BAR{C} \}
                }
                {
                \vlhy{}
                }
            }
        }
        }
    $$ where $*$ denotes where the inductive hypothesis is used.
\end{enumerate}

For $\OR^{-1}$: 
suppose there is a proof $\pi$ of $\Gamma  \{ A \OR B\}$.
We show the rule $\OR^{-1}$ is height and cut-rank preserving admissible by induction on the height of the proof $\pi$.
For the base case $h(\pi)=0$, $\pi$ is of the form $$\vlinf{\mathsf{id}}{}{\Gamma  \{ A \OR B\} \{ p, \BAR{p} \}}{}$$ for some $p \in \mathsf{Atm}$. We have the following proof: $$\vlinf{\mathsf{id}}{}{\Gamma  \{ A, B\} \{ p, \BAR{p} \}}{}$$ where we note the height and cut-rank is preserved.

Assume $\OR^{-1}$ is height-preserving and cut-rank admissible for proofs $\pi$ of $h(\pi) < h$ for some natural number $h>0$.

When $h(\pi)=h$, we look at cases:
\begin{itemize}
    \item Case I: $\pi$ is of the form:
    $$
    \vlderivation
    {
    \vlin{\OR}{}
    {
        \Gamma  \{ A \OR B\}
    }
    {
    \vlhtr{\pi'}{\Gamma  \{ A, B\}}
    }
    }
    $$
    for some proof $\pi'$. Then we have:
    $$
    \vlderivation{
        \vlhtr{\pi'}{\Gamma  \{ A, B\}}
        }
    $$
    where cuts of higher rank are not introduced and height is preserved with $h(\pi') = h(\pi) - 1 \leq h(\pi)$.
    \item Case II: $\pi$ is of the form:
    $$
    \vlderivation
    {
        \vliiin{\rho}{}
        {
            \Gamma \{ A \OR B \} 
        }
        {
            \vlhtr{\pi_1}{\Gamma_1 \{ A \OR B \}}
        }
        {
            \vlhy{\dots}
        }
        {
            \vlhtr{\pi_m}{\Gamma_m \{ A \OR B \}}
        }
    }
    $$
    for some rule $\rho$, some proofs $\pi_1, \dots , \pi_m$. Then we have a proof $\pi'$ of $\seq{\Gamma}{A,B}$ using the inductive hypothesis on proofs $\pi_1, \dots , \pi_m$:
    $$
    \vlderivation
    {
        \vliiin{\rho}{}
        {
            \Gamma \{ A , B \}
        }
        {
            \vlin{\OR^{-1}}{*}
            {
                \Gamma_1 \{ A , B \}
            }
            {\vlhtr{\pi_1}{\Gamma_1 \{ A \OR B \} }}
        }
        {
            \vlhy{\dots}
        }
        {
            \vlin{\OR^{-1}}{*}
            {
                \Gamma_m \{ A , B \} 
            }
            {\vlhtr{\pi_m}{\Gamma_m \{ A \OR B \}}}
        }
    }
    $$
    where $*$ denotes where the inductive hypothesis is used. Here cuts of higher rank are not introduced and the height is preserved - the heights of the proofs of $\Gamma_j \{A, B \}$ are $h(\pi_j)$ as $\OR^{-1}$ is height preserving by the inductive hypothesis, and so $h(\pi')=1+\text{max}(h(\pi_1), \dots, h(\pi_m)) = h(\pi)$.
    \end{itemize}

    For $\AND^{-1}_i$ and $\BOX^{-1}$: the proofs of the height and cut-rank preserving admissibility are identical to the previous one for $\OR^{-1}$.

    For $\mathsf{wk}$ and $\mathsf{cont}$: the proofs of the height and cut-rank preserving admissibility can be adapted from those in~\cite{brunnler_deep_2009} without any special treatment of the additional rules in $\DIA_{\mathsf{k}{\mathsf{X}}}$.

    The inverse rules of $\DIA$, $\DIA_{\mathsf{k}j}$ and $\mathsf{cut}$ are cases of weakening from which we can also deduce is cut-rank and height-preserving admissible.
\end{proof}

To conclude this section we state the completeness of $\mathsf{nK} + \DIA_{\mathsf{k}{\mathsf{X}}} + \mathsf{cut}$ with respect to the quasi-transitive modal logics.

\begin{theorem} \label{thm:cutcompleteness}
    Let $\mathsf{X} \subseteq  \{n \in \mathbb{N}: n>1 \}$. Let $A$ be a formula. If $\mathsf{K} + \path{\mathsf{X}} \vdash A$, then $\mathsf{nK} + \DIA_{\mathsf{kX}} + \mathsf{cut} \vdash A$.
\end{theorem}
\begin{proof}
Knowing from~\cite{brunnler_deep_2009} that the axioms and rules of $\K$ are derivable in $\mathsf{nK} + \mathsf{cut}$, we only need to show that for any $n \ge 1$, axiom $\path{n}$ is derivable using the rule $\DIA_{\mathsf{k}n}$. 
    $$
    \vlderivation
    {
        \vlin{\OR}{}
        {
            \DIA \BAR{A} \OR \BOX^n A
        }
        {
            \vlin{\BOX}{\text{\footnotesize $n$ times}}
            {
                \DIA \BAR{A} , \BOX^n A
            }
            {
                \vlin{\DIA_{\mathsf{k}n}}{}
                {
                    \DIA \BAR{A} , [\dots [A] \dots]
                }
                {
                    \vlin{\mathsf{gid}}{}
                    {
                        \DIA \BAR{A} , [\dots [\BAR{A}, A] \dots]
                    }
                    {
                        \vlhy{}
                    }
                }
            }
        }
    }
    $$
    We can conclude by admissibility of $\mathsf{gid}$ (Proposition~\ref{prop:invadm}).
\end{proof}

Note that the system $\mathsf{nK} + \DIA_{\mathsf{k}{\mathsf{X}}}$ is modularly complete \emph{with $\mathsf{cut}$}.
This is in contrast with the systems from~\cite{gore_correspondence_2011} which are cut-free complete but require a notion of completion of the rules wrt.~the set of axioms, and therefore cannot be considered modular. 
Next, we will show exactly how the necessity for the completion arises in the cut-elimination procedure.

\section{Syntactic Cut-elimination}\label{sec:cutelim}
\subsection{System Completion and Structural rules}
Given a set of quasi-transitivity axioms, the systems proposed in~\cite{gore_correspondence_2011} require to complete this set with additional quasi-transitivity axioms.  
We are able to simplify their notion of completion, originally defined using the algebra of the propagation graphs of nested sequents, as we are exclusively working with quasi-transitivity axioms rather than generic path axioms.

\begin{definition}\label{def:completion}
	Let $\mathsf{X} \subseteq  \{n \in \mathbb{N}: n>1 \}$. 
	We define the \emph{completion} of $\mathsf{X}$, denoted $\Hat{\mathsf{X}}$, inductively as follows: for $p \in \mathbb{N}$, define
	\begin{align*}
		\mathsf{X}_0 &\colonequals \mathsf{X} \\
		\mathsf{X}_{p+1} &\colonequals \mathsf{X}_p \cup \{m+n-1 :m,n \in \mathsf{X}_n\} 
        \\
		\Hat{\mathsf{X}} &\colonequals \bigcup_{p=0}^{\infty} \mathsf{X}_p
	\end{align*}
\end{definition}

We will utilise the following results.
\begin{proposition}
	Let $\mathsf{X} \subseteq  \{n \in \mathbb{N}: n>1 \}$. If $m,n \in \Hat{\mathsf{X}}$, then $m+n-1 \in \Hat{\mathsf{X}}$.
\end{proposition}
\begin{proof}
	The proof follows from the definition of the system completion of $\mathsf{X}$. 
\end{proof}

For the syntactic cut-elimination procedure, we need structural rules corresponding to the quasi-transitivity axioms 
which we introduce in Figure~\ref{fig:struct}. 
\begin{figure}
	\centering
	$
	\vlinf{\boxtimes_{\mathsf{k}n}}{}
	{
		\seq{\Gamma}{[\Delta_1, [\dots,[\Delta_{n-1},[\Delta_n, \Delta]] \dots]]}
	}
	{
		\seq{\Gamma}{[\Delta], [\Delta_1, [\dots,[\Delta_{n-1},[\Delta_j]] \dots]]}
	}
	$
	\caption{Modal structural rule $\boxtimes_{\mathsf{k}n}$ where $n \geq 1$}
	\label{fig:struct}
\end{figure}
Given $\mathsf{X} \subseteq  \{  n \in \mathbb{N} : n>1 \}$, we write $\boxtimes_{\mathsf{k}\mathsf{X}} \colonequals \{ \boxtimes_{\mathsf{k}n} : n \in \mathsf{X} \}$.

\begin{proposition}\label{prop:smallstrucadm}
    Let $\mathsf{X} \subseteq  \{n \in \mathbb{N}: n>1 \}$. Each rule in $\boxtimes_{\mathsf{k}\mathsf{X}}$ is cut-rank admissible in $\mathsf{nK} + \DIA_{\mathsf{k}\hat{\mathsf{X}}} + \mathsf{cut}$.
\end{proposition}
\begin{proof}
    Let us write 
    $$\vlinf{\boxtimes_{\mathsf{k}n}'}{}
        {
            \Gamma  \{ [ \dots [\Delta] \dots ]  \}
        }
        {
        \Gamma  \{ [\Delta]  \}
        }$$
    with $\text{depth}(\Gamma  \{ [ \dots [ \{ \  \}] \dots ]  \})=n$ for a simplified instance of $\boxtimes_{\mathsf{k}n}$ where all $\Delta_i$'s are empty (and weakening is built in).
    Suppose there is a proof $\pi$ of $\Gamma  \{ [\Delta]  \}$.
    We first show each $\boxtimes_{\mathsf{k}n}'$ is cut-rank admissible by induction on the height of the proof $\pi$.

    Base case $h(\pi) = 0$, $\pi$ is of the form:
    $$
    \begin{array}{ccc}
        \vlinf{\mathsf{id}}{}{\seq{\Gamma}{[\seq{\Delta}{p, \Bar{p}}]}}{} & \text{or} & \vlinf{\mathsf{id}}{}{\seq{\seq{\Gamma}{[\Delta]}}{p , \Bar{p}}}{}
    \end{array}
    $$
	for some $p \in \mathsf{Atm}$. In either case, we have the following proofs:
	$$
	\begin{array}{ccc}
		\vlinf{\mathsf{id}}{}{\seq{\Gamma}{[ \dots [\seq{\Delta}{p, \Bar{p}}] \dots ]}}{} & \text{or} & \vlinf{\mathsf{id}}{}{\seq{\seq{\Gamma}{[ \dots [\Delta] \dots ]}}{p , \Bar{p}}}{}
	\end{array}
	$$
        where we note cuts have not been introduced.
        
	Assume $\boxtimes_{\mathsf{k}n}'$ is cut-rank admissible for proofs $\pi$ of $h(\pi) < h$ for some natural number $h>0$.
	When $h(\pi)=h$, we look at the following cases:
    \begin{itemize}
        \item Case I: $\pi$ is of the form:
        $$
        \vlderivation
        	{
        		\vliiin{\rho}{}
        		{
        			\seq{\Gamma}{[\Delta]}
        		}
        		{
        			\vlhtr{\pi_1}{\seq{\Gamma_1}{[\Delta_1]}}
        		}
        		{
        			\vlhy{\dots}
 				}
        		{
        			\vlhtr{\pi_l}{\seq{\Gamma_l}{[\Delta_l]}}
        		}
        	}
        $$
        for some rule $\rho$, some proofs $\pi_1, \dots , \pi_l$. Then we have
        $$
             	\vlderivation
        	{
        		\vliiin{\rho}{}
        		{
        			\seq{\Gamma}{[\dots [\Delta] \dots]}
        		}
        		{
        			\vlin{\boxtimes_{\mathsf{k}n}'}{*}
        			{
        				\seq{\Gamma_1}{[\dots [\Delta_1] \dots]}
        			}
        			{
        				\vlhtr{\pi_1}{\seq{\Gamma_1}{[\Delta_1]}}
        			}
        		}	
        		{
        			\vlhy{\dots}
        		}
        		{
        			\vlin{\boxtimes_{\mathsf{k}n}'}{*}
        			{
        				\seq{\Gamma_l}{[\dots [\Delta_l] \dots]}
        			}
        			{
        				\vlhtr{\pi_l}{\seq{\Gamma_l}{[\Delta_l]}}
        			}
        		}
        	}
        $$ 
        where $*$ denotes where we have used the inductive hypothesis on $h(\pi_i) \leq \text{max}(h(\pi_1), \dots, h(\pi_m)) < h(\pi)=h$. Here, the cut-rank has been preserved.
        \item Case II: $\pi$ is of the form:
        $$
        \vlderivation
        	{
        		\vlin{\DIA}{}
        		{
        			\seq{\Gamma}{\DIA A, [\Delta]}
        		}
        		{
        			\vlhtr{\pi'}{\seq{\Gamma}{\DIA A, [A, \Delta]}}
        		}
        	}
        $$
        Then we have
        $$
                	\vlderivation
        	{
        		\vlin{\DIA_{\mathsf{k}n}}{}
        		{
        			\seq{\Gamma}{\DIA A, [\dots [\Delta] \dots ]}
        		}
        		{
        			\vlin{\boxtimes_{\mathsf{k}n}'}{*}
        			{
        				\seq{\Gamma}{\DIA A, [ \dots [A, \Delta] \dots ]}
        			}
        			{
        				\vlhtr{\pi'}{\seq{\Gamma}{\DIA A, [A, \Delta]}}
        			}
        		}
        	}
        $$
        where $*$ denotes where we have used the inductive hypothesis on $h(\pi') < h(\pi)=h$. Here, additional cuts have not been introduced.
        \item Case III: $\pi$ is of the form:
        $$
        \vlderivation
        	{
        		\vlin{\DIA_{\mathsf{k}m}}{}
        		{
        			\seq{\Gamma}{\DIA A, [\Delta_1, [\dots , [\Delta_{i-1}, [\seq{\Delta_i}{\varnothing}]] \dots ]]}
        		}
        		{
        			\vlhtr{\pi'}{\seq{\Gamma}{\DIA A, [\Delta_1, [\dots , [\Delta_{i-1}, [\seq{\Delta_i}{A}]] \dots ]]}}
        		}
        	}
        $$
        where $\text{depth}(\seq{\Delta_i}{\ }) = m-i$. 
        We note that $$\text{depth}(\seq{\Gamma}{\DIA A, [\Delta_1, [\dots , [\Delta_{i-1}, [ \dots [\seq{\Delta_i}{\ }] \dots ]] \dots ]]}) = m+n-1$$ As $\DIA_{\mathsf{k}m}, \DIA_{\mathsf{k}n} \in \DIA_{\mathsf{k}\mathsf{X}}$, by the definition of completion, $\DIA_{\mathsf{k}(m+n-1)} \in \DIA_{\mathsf{k}\mathsf{X}}$ and we have the following proof:
        $$
        \vlderivation
        {
        	\vlin{\DIA_{\mathsf{k}(m+n-1)}}{}
        	{
        		\seq{\Gamma}{\DIA A, [\Delta_1, [\dots , [\Delta_{i-1}, [ \dots [\seq{\Delta_i}{\varnothing}] \dots ]] \dots ]]}
        	}
        	{
        		\vlin{\boxtimes_{\mathsf{k}n}'}{*}
        		{
        			\seq{\Gamma}{\DIA A, [\Delta_1, [\dots , [\Delta_{i-1}, [ \dots [\seq{\Delta_i}{A}] \dots ]] \dots ]]}
        		}
        		{
        			\vlhtr{\pi'}{ \seq{\Gamma}{\DIA A, [\Delta_1, [\dots , [\Delta_{i-1}, [\seq{\Delta_i}{A}]] \dots ]]}}
        		}
        	}
        }
        $$
        where $*$ denotes where we have used the inductive hypothesis on $h(\pi') < h(\pi) = h$.
    \end{itemize}
	
    We show each $\boxtimes_{\mathsf{k}n}$ is cut-rank admissible by using that each $\boxtimes_{\mathsf{k}n}'$ is cut-rank admissible and using the cut-rank admissibility of $\mathsf{cont}$ and $\mathsf{wk}$.
\end{proof}

The following is a stronger statement which follows from Proposition~\ref{prop:smallstrucadm}.

\begin{proposition} \label{prop:strucadm}
    Let $\mathsf{X} \subseteq  \{n \in \mathbb{N}: n>1 \}$. Each rule in $\boxtimes_{\mathsf{k}\hat{\mathsf{X}}}$ is cut-rank admissible in $\mathsf{nK} + \DIA_{\mathsf{k}\hat{\mathsf{X}}} + \mathsf{cut}$.
\end{proposition}
\begin{proof}
    This is an induction on the structure of $\hat{\mathsf{X}}$. The argument follows from the fact that the structural rule $\boxtimes_{\mathsf{k}(m+l-1)}$ is derivable from $\boxtimes_{\mathsf{k}m}$ and $\boxtimes_{\mathsf{k}l}$ with weakening:
 \begin{equation*}
 \vlderivation
 {
    \vlin{\boxtimes_{\mathsf{k}l}}{}
    {
        \seq{\Gamma}{[\Delta_1, [\dots,[\Delta_{m+l-2},[\Delta_{m+l-1}, \Delta]] \dots]]}
    }
    {
        \vlin{\boxtimes_{\mathsf{k}m} + \mathsf{wk}}{}
        {
            \seq{\Gamma}{[\Delta_1, [\dots,[\Delta_{m}, [\Delta], [\dots, [\Delta_{m+l-2},[\Delta_{m+l-1}, \Delta]]\dots]] \dots]]}
        }
        {
            \vlhy{\seq{\Gamma}{[\Delta], [\Delta_1, [\dots,[\Delta_{m+l-2},[\Delta_{m+l-1}]] \dots]]}}
        }
    }
 }
 \end{equation*}
\end{proof}

\subsection{Cut-elimination Theorem}
In this subsection, we present the cut-elimination procedure. 
where we make use of the admissibility of the modal structural rules to reason about cut-reductions.

\begin{lemma}[Cut reduction]\label{lem:cutred}
    Let $\mathsf{X} \subseteq  \{n \in \mathbb{N}: n>1 \}$. Let $\Gamma$ be a sequent. If $\mathsf{nK} + \DIA_{\mathsf{k}\hat{\mathsf{X}}} + \mathsf{cut}_{r+1} \vdash \Gamma$, then $\mathsf{nK} + \DIA_{\mathsf{k}\hat{\mathsf{X}}} + \mathsf{cut}_{r} \vdash \Gamma$
\end{lemma}
\begin{proof}
    We proceed to prove this result by induction on the number of cuts of rank $r+1$ in a proof. The base case, when there are no cuts of rank $r+1$ in a proof is immediate.

	Assume the result holds for proofs with the number of cuts of rank $r+1$ up to some natural number $s>0$.
	
	Given a proof $\pi$ with $s$ cuts of rank $r+1$, there is a cut of rank $r+1$, such that there is a subproof of $\pi$ of the form
	$$
	\vlderivation
	{
		\vliin{\mathsf{cut}_{r+1}}{}
		{
			\seq{\Gamma}{\varnothing}
		}
		{
			\vlhtr{\pi_1}{\seq{\Gamma}{A}}
		}
		{
			\vlhtr{\pi_2}{\seq{\Gamma}{\BAR{A}}}
		}
	}
	$$
	for some formula $A$ of degree less than or equal to $r+1$, and some proofs $\pi_1$ and $\pi_2$ in $\mathsf{nK} + \DIA_{\mathsf{k}\hat{\mathsf{X}}} + \mathsf{cut}_{r}$. 
    We proceed by induction on $h(\pi_1)+h(\pi_2)$. The base case can be found in \cite{brunnler_deep_2009}. 
    In the case when $\BAR{A}$ is not principal in the last rule applied in $\pi_2$, i.e., that rule does not affect the cut-formula, or is the result of a rule in $\mathsf{nK}+\mathsf{cut}_r$, the cut-reduction can be found in \cite[Lemma 11]{brunnler_deep_2009}. This utilises the fact that the inverses of the rules in $\mathsf{nK} + \mathsf{cut}_r$ are height and cut-rank preserving admissible using Proposition~\ref{prop:invadm}.

    In the case where the final rule in $\pi_2$ is $\DIA_{\mathsf{k}n}$ and the cut is of the form:
    $$
    \vlderivation
        {
            \vliin{\mathsf{cut}_{r+1}}{}
            {
                \Gamma  \{ 
[\Delta_1, [\dots, [\Delta_{n-1},[\Delta_{n}]] \dots ]] \}
            }
            {
                \vlhtr{\pi_1}{\Gamma  \{ 
\BOX \BAR A, [\Delta_1, [\dots, [\Delta_{n-1},[\Delta_{n}]] \dots ]] \}}
            }
            {
                \vlin{\DIA_{\mathsf{k}n}}{}
                {
                \Gamma  \{ 
\DIA A, [\Delta_1, [\dots, [\Delta_{n-1},[\Delta_{n}]] \dots ]] \}
                }
                {
                    \vlhtr{\pi_2'}{\Gamma  \{ 
\DIA A, [\Delta_1, [\dots, [\Delta_{n-1},[A, \Delta_{n}]] \dots ]] \}}
                }
            }
        }
    $$

    We have a proof denoted $\pi_1'$ of $\Gamma  \{ 
[\BAR{A}], [\Delta_1, [\dots, [\Delta_{n-1},[\Delta_{n}]] \dots ]] \}$ in $\mathsf{nK}+\DIA_{\mathsf{k}\hat{\mathsf{X}}}+\mathsf{cut}_r$ which is obtained by applying the cut-rank admissible rule $\BOX^{-1}$ on the proof $\pi_1$ of $\Gamma  \{ 
\BOX \BAR{A}, [\Delta_1, [\dots, [\Delta_{j-1},[\Delta_{j}]] \dots ]] \}$ using Proposition~\ref{prop:invadm}. We note $r(\pi_1') \le r(\pi_1)$.
   
   We then have a proof denoted $\pi_3$ of $\Gamma  \{ 
[\Delta_1, [\dots, [\Delta_{j-1},[\BAR{A}, \Delta_{j}]] \dots ]] \}$ in $\mathsf{nK}+\DIA_{\mathsf{k}\hat{\mathsf{X}}}+\mathsf{cut}_r$, obtained from proof $\pi_1'$ by utilising the cut-rank admissibility of the structural rule~$\boxtimes_{\mathsf{k}n}$ using Proposition~\ref{prop:strucadm}. We note $r(\pi_3) \le r(\pi_1') \le r(\pi_1)$.

     On the other hand, we have a proof denoted $\pi_1''$ of $\Gamma  \{ 
\BOX \BAR{A}, [\Delta_1, [\dots, [\Delta_{n-1},[A, \Delta_{n}]] \dots ]] \}$ in $\mathsf{nK}+\DIA_{\mathsf{k}\hat{\mathsf{X}}}+\mathsf{cut}_r$ which is obtained by applying the cut-rank height-preserving admissible rule $\mathsf{wk}$ on the proof $\pi_1$ of $\Gamma  \{ 
\BOX \BAR{A}, [\Delta_1, [\dots, [\Delta_{n-1},[\Delta_{n}]] \dots ]] \}$ using Proposition~\ref{prop:invadm}. We note $h(\pi_1'')\leq h(\pi_1)$.

    We then get a proof $\pi_4$ of $\Gamma  \{ 
    [\Delta_1, [\dots, [\Delta_{n-1},[\bar A, \Delta_{n}]] \dots ]] \}$ in $\mathsf{nK}+\DIA_{\mathsf{k}\hat{\mathsf{X}}}+\mathsf{cut}_r$ from the proofs $\pi_1''$ and $\pi_2'$ utilising the inductive hypothesis:
    \begin{equation*}
        \vlderivation
    {
        \vliin{\mathsf{cut}_{r+1}}{*}
                {
                    \Gamma  \{ 
[\Delta_1, [\dots, [\Delta_{n-1},[A, \Delta_{n}]] \dots ]] \}
                }
                {
                    \vlhtr{\pi_1''}{\Gamma  \{ 
\BOX \BAR{A}, [\Delta_1, [\dots, [\Delta_{n-1},[A, \Delta_{n}]] \dots ]] \}}
                }
                {
                    \vlhtr{\pi_2'}{\Gamma  \{ \DIA A,
[\Delta_1, [\dots, [\Delta_{n-1},[A, \Delta_{n}]] \dots ]] \}}
                }
    }
    \end{equation*}
  
$*$ denotes where we have used the inductive hypothesis on $h(\pi_1'') + h(\pi_2') < h(\pi_1) + h(\pi_2)$.

And so we have a proof of $\Gamma  \{ 
[\Delta_1, [\dots, [\Delta_{n-1},[\Delta_{n}]] \dots ]] \}$ in $\mathsf{nK}+\DIA_{\mathsf{k}\hat{\mathsf{X}}}+\mathsf{cut}_r$:
        $$
        \vlderivation
        {
            \vliin{\mathsf{cut}_{r}}{}
            {
                \Gamma  \{ 
[\Delta_1, [\dots, [\Delta_{n-1},[\Delta_{n}]] \dots ]] \}
            }
            {
                \vlhtr{\pi_3}{\Gamma  \{ 
[\Delta_1, [\dots, [\Delta_{n-1},[\BAR{A}, \Delta_{n}]] \dots ]] \}}
            }
            {
                \vlhtr{\pi_4}{\Gamma  \{ 
[\Delta_1, [\dots, [\Delta_{n-1},[A, \Delta_{n}]] \dots ]] \}}
            }
        }$$
    
    It is similar when the final rule of $\pi_1$ is $\DIA_{\mathsf{k}n}$.
    
    Applying this transformation on the proof of $\Gamma$ reduces the number of cuts of rank $r+1$ by 1 and we can apply the external inductive hypothesis to achieve the result.
\end{proof}

\begin{theorem}[Cut-elimination] \label{thm:cutelim}
    Let $\mathsf{X} \subseteq  \{n \in \mathbb{N}: n>1 \}$. Let $\Gamma$ be a sequent. If $\mathsf{nK} + \DIA_{\mathsf{k}\hat{\mathsf{X}}} + \mathsf{cut} \vdash \Gamma$, then $\mathsf{nK} + \DIA_{\mathsf{k}\hat{\mathsf{X}}}  \vdash \Gamma$.
\end{theorem}
\begin{proof}
    We proceed by induction on the maximal cut-rank of a proof of $\Gamma$. The base case $0$ is a proof in $\mathsf{nK} + \DIA_{\mathsf{k}\hat{\mathsf{X}}}$. Assume proofs of maximal cut-rank $r$ for some natural number $r>0$ can be reduced to a cut-free one.

    Suppose we have a proof of maximal cut-rank $r+1$. Using Lemma~\ref{lem:cutred}, the proof can be reduced to a proof of maximal cut-rank $r$. Applying the inductive hypothesis reduces the proof to a cut-free proof.
\end{proof}

The following is a result from \cite{gore_correspondence_2011}. The previous theorem provides an alternative proof.

\begin{corollary}[Cut-free completeness]\label{cor:firstresult}
    Let $\mathsf{X} \subseteq  \{n \in \mathbb{N}: n>1 \}$. Let $A$ be a formula. Then, if $\mathsf{K} + \path{\mathsf{X}} \vdash A$, then $\mathsf{nK} + \DIA_{\mathsf{k}\hat{\mathsf{X}}}  \vdash A$.
\end{corollary}
    
\begin{proof}
	This is as a consequence of Proposition~\ref{thm:cutcompleteness} and Theorem~\ref{thm:cutelim}.
\end{proof}

\section{Modularity}\label{sec:modular}
We have given an alternative proof that the systems from~\cite{gore_correspondence_2011} are cut-free complete for quasi-transitive modal logics.
However, they are not modular as, for $\mathsf{X}_1, \mathsf{X}_2 \subseteq  \{ n \in \mathbb{N} : n>1 \}$, the completion of $\mathsf{X}_1 \cup \mathsf{X}_2$ is not in general $\Hat{\mathsf{X}}_1 \cup \Hat{\mathsf{X}}_2$, meaning that one might need to add more rules to capture $\mathsf{4}^{\mathsf{X}_1 \cup \mathsf{X}_2}$ than just the ones occurring in $\DIA_{\mathsf{k}\Hat{\mathsf{X_1}}} \cup \DIA_{\mathsf{k}\Hat{\mathsf{X_2}}}$.

To achieve modularity we use new rules given in Figure~\ref{fig:new}. This rule ``propagate" formula $\DIA A$ through the nested sequent tree.

Given $\mathsf{X} \subseteq  \{ n \in \mathbb{N} : n>1 \}$, denote $\DIA_{\mathsf{4}\mathsf{X}} \colonequals \{ \DIA_{\mathsf{4}n} : n \in \mathsf{X} \}$.

In this new system, we avoid having to prove cut-elimination and we utilise the previous results. We conjecture that a direct cut-elimination would be possible in this system utilising a different cut-rule similar the $\mathsf{4cut}$ used in~\cite{brunnler_deep_2009}.

\begin{figure}
    \centering
    $\vlinf{\DIA_{\mathsf{4}n}}{}
    {\Gamma  \{ 
\DIA A, [\Delta_1, [\dots, [\Delta_{n-2},[\Delta_{n-1}]] \dots ]] \}}
    {\Gamma  \{ 
\DIA A, [\Delta_1, [\dots, [\Delta_{n-2},[\DIA A, \Delta_{n-1}]] \dots ]] \}}$
    \caption{Modal propagation rules $\DIA_{\mathsf{4}n}$ where $n>1$}
    \label{fig:new}
\end{figure}

\begin{proposition}
    Let $\mathsf{X} \subseteq  \{n \in \mathbb{N}: n>1 \}$. The $\mathsf{wk}$ rule is admissible in $\mathsf{nK} + \DIA_{\mathsf{4}\mathsf{X}}$.
\end{proposition}
\begin{proof}
    The proof follows the one in~\cite{brunnler_deep_2009}
    with no particular adaptation due to the additional modal propagation rules.
\end{proof}

\begin{proposition}\label{prop:translation}
    Let $\mathsf{X} \subseteq  \{n \in \mathbb{N}: n>1 \}$. Let $\Gamma$ be a sequent. If $\mathsf{nK} + \DIA_{\mathsf{k}\mathsf{X}} \vdash \Gamma$, then $\mathsf{nK} + \DIA_{\mathsf{4}\mathsf{X}}  \vdash \Gamma$.
\end{proposition}
\begin{proof}
    Each $\DIA_{\mathsf{k}n} \in \DIA_{\mathsf{k}\mathsf{X}}$ is derivable 
    in $\mathsf{nK} + \DIA_{\mathsf{4}\mathsf{X}}$:
    \begin{equation*}
\vlderivation
        {
            \vlin{\DIA_{\mathsf{4}n}}{}
            {
                \Gamma  \{ 
\DIA A, [\Delta_1, [\dots, [\Delta_{n-1},[\Delta_{n}]] \dots ]] \}   
            }
            {
                \vlin{\DIA}{}
                {
                    \Gamma  \{ 
\DIA A, [\Delta_1, [\dots, [\Delta_{n-1}, \DIA A, [\Delta_{n}]] \dots ]] \}
                }
                {
                    \vlin{\mathsf{wk}}{}
                    {
                        \Gamma  \{ 
\DIA A, [\Delta_1, [\dots, [\Delta_{n-1}, \DIA A, [A, \Delta_{n}]] \dots ]] \}
                    }
                    {
                        \vlhy{\Gamma  \{ 
\DIA A, [\Delta_1, [\dots, [\Delta_{n-1}, [A, \Delta_{n}]] \dots ]] \}}
                    }
                }
            }
        }
    \end{equation*}
\end{proof}

\begin{lemma}[System completion reduction]\label{lem:sysred}
     Let $\mathsf{X} \subseteq  \{n \in \mathbb{N}: n>1 \}$. Let $\Gamma$ be a sequent. If $\mathsf{nK} + \DIA_{\mathsf{4}\hat{\mathsf{X}}} \vdash \Gamma$, then $\mathsf{nK} + \DIA_{\mathsf{4}\mathsf{X}}  \vdash \Gamma$
\end{lemma}
\begin{proof}
    We have $\Hat{\mathsf{X}} = \bigcup_{p=0}^{\infty} \mathsf{X}_p$ as given in Definition~\ref{def:completion}. We show each $\DIA_{\mathsf{4}n} \in \DIA_{\mathsf{4}\hat{\mathsf{X}}}$ is admissible in $\mathsf{nK} + \DIA_{\mathsf{4}\mathsf{X}}$. By definition, $\DIA_{\mathsf{4}n} \in \DIA_{\mathsf{4}\mathsf{X}_p}$ for some $p \in \mathbb{N}$. We show each $\DIA_{\mathsf{4}n} \in \DIA_{\mathsf{4}\mathsf{X}_p}$ is admissible $\mathsf{nK} + \DIA_{\mathsf{4}\mathsf{X}}$ by induction on $p$.

    The base case $p=0$ follows from the definition of $\mathsf{X}_0 = \mathsf{X}$.

    Assume the result holds for some $p \in \mathbb{N}$.

    Suppose $\DIA_{\mathsf{4}n} \in \DIA_{\mathsf{4}\mathsf{X}_{p+1}}$. By definition, we have $n \in \mathsf{X}_{p}$, or $n = m+l-1$ for some $m, l \in \mathsf{X}_p$.

    If $n \in \mathsf{X}_{p}$: $\DIA_{\mathsf{4}n} \in \DIA_{\mathsf{4}\mathsf{X}_{p}}$ and the result follows from the induction hypothesis.

    If $n = m+l-1$ for some $m, l \in \mathsf{X}_n$: As $\DIA_{\mathsf{4}m}, \DIA_{\mathsf{4}l} \in \DIA_{\mathsf{4}\mathsf{X}_p}$, they are admissible in $\mathsf{nK} + \DIA_{\mathsf{4}\mathsf{X}}$ by induction hypothesis, and we see the admissibility of $\DIA_{\mathsf{4}n}=\DIA_{\mathsf{4}(m+l-1)}$ in $\mathsf{nK} + \DIA_{\mathsf{4}\mathsf{X}}$ through the following derivation:
    \begin{equation*}
\vlderivation
{
    \vlin{\DIA_{\mathsf{4}l}}{*}
    {
        \Gamma  \{ 
\DIA A, [\Delta_1, [\dots, [\Delta_{m+l-3},[\Delta_{m+l-2}]] \dots ]] \}
    }
    {
        \vlin{\DIA_{\mathsf{4}m}}{*}
        {
            \Gamma  \{ 
\DIA A, [\Delta_1, [\dots, [\Delta_{l-1}, \DIA A,[\dots,[\Delta_{m+l-3},[\Delta_{m+l-2}]]]\dots] \dots ]] \}
        }
        {
            \vlin{\mathsf{wk}}{}
            {
            \Gamma  \{ 
\DIA A, [\Delta_1, [\dots, [\Delta_{l-1},\DIA A,[\dots,[\Delta_{m+l-3},[\DIA A, \Delta_{m+l-2}]]]\dots] \dots ]] \}
            }
            {
            \vlhy{\Gamma  \{ 
\DIA A, [\Delta_1, [\dots, [\Delta_{l-1}, [\dots,[\Delta_{m+l-3},[\DIA A, \Delta_{m+l-2}]]]\dots] \dots ]] \}}
            }
        }
    }
}
    \end{equation*}
    where $*$ denotes using the inductive hypothesis of the admissibility of $\DIA_{\mathsf{4}m}$ and $\DIA_{\mathsf{4}l}$.
\end{proof}

We require the following technical lemma to prove soundness.

\begin{lemma}\label{lem:soundnesslemma}
    Let $\mathsf{X} \subseteq  \{n \in \mathbb{N}: n>1 \}$. Let $\Gamma \{ \ \}$ be a context. Let $A$ and $B$ be formulas. If $\mathsf{K} + \path{\mathsf{X}} \vdash A \IMP B$, then $\mathsf{K} + \path{\mathsf{X}} \vdash \form(\Gamma \{ A \}) \IMP \form(\Gamma \{ B \})$.
\end{lemma}

\begin{proof}
    The proof follows the same steps as in~\cite[Lemma 3]{brunnler_deep_2009}.
\end{proof}

We conclude with soundness and cut-free completeness for a modular system for quasi-transitive modal logics.

\begin{theorem}\label{thm:modular}
    Let $\mathsf{X} \subseteq  \{n \in \mathbb{N}: n>1 \}$. Let $A$ be a formula. Then, 
    \begin{enumerate*}
    \item $\mathsf{K} + \path{\mathsf{X}} \vdash A$ iff 
    \item $\mathsf{nK} + \DIA_{\mathsf{4}\mathsf{X}}  \vdash A$.
    \end{enumerate*}
\end{theorem}
\begin{proof}
    $1 \Rightarrow 2$: By Corollary \ref{cor:firstresult} we have $\mathsf{nK} + \DIA_{\mathsf{k}\hat{\mathsf{X}}}  \vdash A$. By Proposition~\ref{prop:translation}, we have $\mathsf{nK} + \DIA_{\mathsf{4}\hat{\mathsf{X}}}  \vdash A$. The result follows from Lemma~\ref{lem:sysred}.

    $2 \Rightarrow 1$ requires to check the rules of $\mathsf{nK} + \DIA_{\mathsf{4}\mathsf{X}}$ are sound. Soundness of the rules of $\mathsf{nK}$ has been proven in~\cite{brunnler_deep_2009}. We prove the soundness of each rule $\DIA_{\mathsf{4}n} \in \DIA_{\mathsf{4}\mathsf{X}}$. 
    By Lemma~\ref{lem:soundnesslemma}, it suffices to show 
    \begin{multline*}
        \mathsf{K} + \path{\mathsf{X}} \vdash \form(\DIA A, [\Delta_1, [\dots, [\Delta_{n-2},[\DIA A, \Delta_{n-1}]] \dots ]])
        \\ \IMPLIES \form(\DIA A, [\Delta_1, [\dots, [\Delta_{n-2},[\Delta_{n-1}]] \dots ]])        \end{multline*}

    We utilise the soundness of $\DIA+ \mathsf{wk}$ with the following derivation tree:
    $$
    \vlderivation
    {
        \vlin{\DIA+ \mathsf{wk}}{}
        {
            \DIA A, \DIA^n A, [\Delta_1, [\dots, [\Delta_{n-2},[\Delta_{n-1}]] \dots ]]
        }
        {
            \vlin{\DIA+ \mathsf{wk}}{}
            {
                \DIA A, [\DIA^{n-1} A, \Delta_1, [\dots, [\Delta_{n-2},[\Delta_{n-1}]] \dots ]]
            }
            {
                \vlin{\DIA+ \mathsf{wk}}{}
                {
                    \vdots
                }
                {
                    \vlin{\DIA+ \mathsf{wk}}{}
                    {
                        \DIA A, [ \Delta_1, [\dots, [\DIA^{2} A, \Delta_{n-2},[\Delta_{n-1}]] \dots ]]
                    }
                    {
                        \vlhy{\DIA A, [ \Delta_1, [\dots, [ \Delta_{n-2},[\DIA A, \Delta_{n-1}]] \dots ]]}
                    }
                }
            }
        }
    }
    $$
    to deduce 
       \begin{multline*}
        \mathsf{K} + \path{\mathsf{X}} \vdash \form(\DIA A, [ \Delta_1, [\dots, [ \Delta_{n-2},[\DIA A, \Delta_{n-1}]] \dots ]])
        \\ \IMPLIES \form(\DIA A, \DIA^n A, [\Delta_1, [\dots, [\Delta_{n-2},[\Delta_{n-1}]] \dots ]])
    \end{multline*}
    By $\path{n}$: $\DIA^n A \IMPLIES \DIA A$, we have
    \begin{multline*}
        \mathsf{K} + \path{\mathsf{X}} \vdash \form(\DIA A, \DIA^n A, [\Delta_1, [\dots, [\Delta_{n-2},[\Delta_{n-1}]] \dots ]])
        \\ \IMPLIES \form(\DIA A, [\Delta_1, [\dots, [\Delta_{n-2},[\Delta_{n-1}]] \dots ]])
    \end{multline*}
    Applying the transitivity of $\IMPLIES$ we have the desired result.
\end{proof}

\bibliography{Library}
\bibliographystyle{plain}

\end{document}